\documentclass[runningheads]{llncs}
\usepackage{graphicx}
\usepackage{adjustbox}

\usepackage{xspace}
\usepackage{hyperref}
\usepackage{listings}
\usepackage{amsmath}
\usepackage{pifont}
\usepackage{amssymb}
\usepackage{color}
\usepackage{subfig}
\usepackage{wrapfig}
\usepackage{float}
\usepackage{float}
\usepackage{refs}
\usepackage{graphicx}
\usepackage{pgfplots}
\usepackage{pgfplotstable}
\usepackage{verbatim,bm}
\usepackage{stmaryrd}
\usepackage{multirow}
\usepackage{array}
\usepackage{tikz}
\usepackage{tabularx}
\usepackage{longtable}
\usepackage{arydshln}
\usetikzlibrary{arrows,positioning}
\usepackage[T1]{fontenc}

\makeatletter
\g@addto@macro\normalsize{%
  \setlength\abovedisplayskip{4pt}
  \setlength\belowdisplayskip{3pt}
  \setlength\abovedisplayshortskip{3pt}
  \setlength\belowdisplayshortskip{3pt}
}
\makeatother
\usepackage[breakable, theorems, skins]{tcolorbox}
\newenvironment{mybox}[1][gray!20]{
    \begin{tcolorbox}[   
        breakable,
        left=0pt,
        right=0pt,
        top=0pt,
        bottom=-1pt,
        colback=#1,
        colframe=#1,
        width=\dimexpr\textwidth\relax,
        boxsep=4pt,
        arc=0pt,outer arc=0pt,    
    ]
}{
    \end{tcolorbox}
}

\newcounter{resq}
\newenvironment{resq}[1][]
{
	\vspace{-2mm}
    \refstepcounter{resq}\par\medskip    
    \begin{mybox}
    \noindent \textbf{EQ~\theresq. #1} \rmfamily
}{
    \end{mybox}
\vspace{-2.5mm}
    \medskip
}

\newcommand{\rone}{(\emph{i})~}
\newcommand{\rtwo}{(\emph{ii})~}
\newcommand{\rthree}{(\emph{iii})~}

\definecolor{mGreen}{rgb}{0,0.6,0}
\definecolor{mGray}{rgb}{0.5,0.5,0.5}
\definecolor{mPurple}{rgb}{0.58,0,0.82}

\lstdefinestyle{CStyle}{
	commentstyle=\color{mGreen},
	keywordstyle=\color{mGreen},
	numberstyle=\tiny\color{mGray},
	escapeinside={(*@}{@*)},
	stringstyle=\color{mPurple},
	basicstyle=\footnotesize\ttfamily,
	breakatwhitespace=false,         
	breaklines=true,                 
	captionpos=b,                    
	keepspaces=true,                 
	numbers=left,                    
	numbersep=5pt,                  
	showspaces=false,                
	showstringspaces=false,
	showtabs=false,                  
	tabsize=2,
	language=C
}

\newcommand {\loris}[1]{{\color{blue}\bf{L: #1}\normalfont}}

\newcommand{\mypar}[1]{\vspace{0.5mm}\noindent  \textit{#1}}

\newcommand{\sem}[1]{\llbracket #1 \rrbracket\xspace}

\newcommand{\en}{\textit{enc}\xspace}
\newcommand{\enprod}{\texttt{En}_\delta\xspace}

\newcommand\eqdef{\stackrel{{\normalfont\mbox{\tiny{def}}}}{=}}

\newcommand{\type}{\tau}
\newcommand{\eq}{eq}
\newcommand{\re}{re}
\newcommand{\xmark}{\ding{55}}%

\def\name{\textsc{nope}\xspace}
\def\sygus{\textsc{SyGuS}\xspace}
\def\qsygus{\textsc{QSyGuS}\xspace}
\def\seahorn{\textsc{SeaHorn}\xspace}
\def\z3{\textsc{Z3}\xspace}
\def\cex{{E}\xspace}
\def\cegis{\textsc{CEGIS}\xspace}
\newcommand\nd[1]{\texttt{nd(#1)}}

\newcommand\nextbit[1]{\texttt{nextbit(#1)\xspace}}

\def\numSolved{59\xspace}
\def\numBenchmarks{132\xspace}
\def\numIfBenchmarks{57\xspace}
\def\numPlusBenchmarks{30\xspace}
\def\numConstBenchmarks{45\xspace}

\def\avgtime{15.59\xspace}

\def\nat{\mathbb{N}}
\def\sy{{sy}}
\def\num{{\#}}
\def\numn{{\#_{nd}}}
\def\produ{pr}
\def\node{q}

\newcommand\pge[2]{P[{#1},{#2}]}

\newcommand\gv[2]{\texttt{g\_#1\_#2}}
\newcommand\gva[3]{\texttt{#1\_#2\_#3}}
\newcommand\child[2]{\texttt{child\_#1\_#2}}

\newcommand{\finding}[1]{\emph{#1}}

\newcommand{\Omit}[1]{}

\usepackage[leftbars,color]{changebar}

\usepackage{ulem}
\cbcolor{red}

\newif\iffull
\fulltrue
\begin{document}
\normalem
\title{Proving Unrealizability \\ for Syntax-Guided Synthesis}
%
%
\author{Qinheping Hu$^1$,  Jason Breck$^1$, John Cyphert$^1$,\\ Loris D'Antoni$^1$, Thomas Reps$^{1,2}$}
\institute{$^1$ University of Wisconsin-Madison, Madison, USA\\
	$^2$ GrammaTech, Inc., USA}
%
%
%
\maketitle              
\begin{abstract}
We consider the problem of automatically establishing that a given
syntax-guided-synthesis (\sygus) problem is \emph{unrealizable} (i.e., has no solution).
Existing techniques have quite limited ability to establish unrealizability for general \sygus instances in which the grammar
	describing the search
	space contains infinitely many programs. 
By encoding the synthesis problem's grammar $G$ as a nondeterministic program $P_G$,
we reduce the unrealizability problem to a reachability problem such that,
if a standard program-analysis tool can establish that a certain
assertion in $P_G$ always holds, then the synthesis problem is unrealizable.

\hspace*{1.5ex}
Our method can be used to augment existing \sygus tools
so that they can establish that a successfully synthesized program $q$
is \emph{optimal} with respect to some syntactic cost---e.g., $q$ has
the fewest possible if-then-else operators.
Using known techniques, grammar $G$ can be transformed
to generate 
the set of
 all programs with lower costs than $q$---e.g.,
fewer conditional expressions.
Our algorithm can then be applied to show that the resulting
synthesis problem is unrealizable.
We implemented the proposed technique in a tool called \name.
\name can prove unrealizability for  \numSolved/\numBenchmarks
variants of existing linear-integer-arithmetic \sygus benchmarks, whereas all existing \sygus solvers lack the ability to prove that
these benchmarks are unrealizable, and time out on them.

\end{abstract}


\section{Introduction}
\label{Se:Introduction}

The goal of program synthesis is to find a program in some search
space that meets a specification---e.g., satisfies a set of examples or a
logical formula.
Recently, a large family of synthesis problems has been unified into a framework called
\emph{syntax-guided synthesis} (\sygus).
A \sygus problem is specified by a regular-tree grammar that describes
the search space of programs,
and a logical formula that constitutes the behavioral specification.
Many synthesizers now support a specific format for \sygus problems~\cite{sygus},
and compete in annual synthesis competitions~\cite{alur2016sygus}.
Thanks to these competitions, these solvers are now quite mature and
are finding a wealth of applications~\cite{pldi17inversion}.

Consider the \sygus problem to synthesize a function $f$ that computes the maximum of two
variables $x$ and $y$, denoted by $(\psi_{\texttt{max2}}(f,x,y), G_1)$.
The goal is to create $e_f$---an expression-tree for $f$---where $e_f$
is in the language of the following regular-tree grammar $G_1$:
\[\hspace{-5pt}
	\begin{array}{r@{~}r@{~~}ll}
	\textrm{Start}	&	::=	& \textrm{Plus}(\textrm{Start}, \textrm{Start}) \hspace{0.34mm}
                                             \mid \textrm{IfThenElse}(\textrm{BExpr}, \textrm{Start}, \textrm{Start})
                                             \mid x \mid y \mid 0 \mid 1\\			
	\textrm{BExpr}	&	::=	& \textrm{GreaterThan}(\textrm{Start}, \textrm{Start})
                                             \mid \textrm{Not}(\textrm{BExpr})
                                             \mid \textrm{And}(\textrm{BExpr}, \textrm{BExpr})
	\end{array}
\]
and $\forall x,y. \psi_{\texttt{max2}}(\sem{e_f},x,y)$ is valid, where 
$\sem{e_f}$ denotes the meaning of $e_f$, and
\[
  \psi_{\texttt{max2}}(f,x,y) := f(x,y)\ge x \wedge f(x,y)\ge y\wedge (f(x,y)=x\vee f(x,y) =y).
\]
\sygus solvers can easily find a solution, such as 
\[
  e := \textrm{IfThenElse}(\textrm{GreaterThan}(x,y),x,y).
\]


Although many solvers can now find solutions efficiently to many \sygus problems,
there has been effectively no work on the much harder task of proving that
a given \sygus problem is \emph{unrealizable}---i.e., it does not admit a solution.
For example, consider the \sygus problem $(\psi_{\texttt{max2}}(f,x,y), G_2)$,
where $G_2$ is the more restricted grammar with if-then-else operators and conditions
stripped out:
\[\hspace{-5pt}
	\begin{array}{r@{~}r@{~~}ll}
	\textrm{Start}	&	::=	& \textrm{Plus}(\textrm{Start}, \textrm{Start}) \hspace{0.34mm}
                                             \mid x \mid y \mid 0 \mid 1
	\end{array}
\]
This \sygus problem does \emph{not} have a solution, because no expression
generated by $G_2$ meets the specification.\footnote{Grammar $G_2$ only generates terms that are equivalent to some linear function of
  $x$ and $y$;
  however, the maximum function cannot be described by a linear function. 
}
However, to the best of our knowledge, current \sygus solvers cannot
prove that such a \sygus problem is unrealizable.\footnote{\Omit{The reader
  might wonder how and why a \sygus problem with such a
  restricted grammar would ever arise.
  }The synthesis problem presented above is one that is generated by a recent tool
  called \qsygus, which extends \sygus with quantitative syntactic objectives~\cite{HuD18}.
  The advantage of using quantitative objectives in synthesis is that they can be
  used to produce higher-quality solutions---e.g., smaller, more readable, more efficient, etc.
  The synthesis problem $(\psi_{\texttt{max2}}(f,x,y), G_2)$ arises from a
  \qsygus problem in which the goal is to produce an expression that
  (i) satisfies the specification $\psi_{\texttt{max2}}(f,x,y)$,
  and (ii) uses the smallest possible number of if-then-else operators. 
  Existing \sygus solvers can easily produce a solution that uses
  one if-then-else operator, but cannot prove that no better solution
  exists---i.e., $(\psi_{\texttt{max2}}(f,x,y), G_2)$ is unrealizable.
  \Omit{What is missing from \qsygus (and other \sygus solvers) is the ability
  to prove that no expression exists for $f$ that uses zero if-then-else
  operators and also satisfies the specification
  $\psi_{\texttt{max2}}(f,x,y)$.
  In this case, we are interested establishing that the synthesis problem
  $(\psi_{\texttt{max2}}(f,x,y), G_2)$ is unrealizable.
  More generally, we wish to develop techniques capable of establishing
  that a given synthesis problem is unrealizable.}
}

	A key property of the previous example is that the grammar is infinite. 
	When such a \sygus problem is realizable, any search technique that
	systematically explores the infinite search space of possible programs
	will eventually identify a solution to the synthesis problem. 
	In contrast, proving that a problem is unrealizable requires showing
	that \emph{every} program in the \emph{infinite} search space
	\emph{fails to
		satisfy} the specification.
	This problem is in general undecidable~\cite{CaulfieldRST15}. Although we cannot hope to have an algorithm for establishing unrealizability, the challenge is to find a technique that succeeds for the kinds of problems encountered in practice.
	\Omit{ However, enumeration search---the technique used in existing solvers---is ill-suited for this problem.  }
	Existing synthesizers can detect the absence of a solution in certain cases (e.g., because the grammar is finite, or is infinite but only generate a finite number of functionally distinct programs).  However, in practice, as our experiments show, this ability is limited---no existing solver was able to show unrealizability for any of the examples considered in this paper.


In this paper, we present a technique for proving that a possibly infinite \sygus problem is unrealizable.
Our technique builds on two ideas.
\begin{enumerate}
  \item
    \label{It:FiniteNumberOfExamples}
    We observe that unrealizability can often be proven using \emph{finitely
    many input examples}.
    In \sectref{IllustrativeExample}, we show how the example
    discussed above can be proven to be unrealizable using four input
    examples---$(0,0)$, $(0,1)$, $(1,0)$, and $(1,1)$.
  \item
    \label{It:EncodingOfGrammarAndExamples}
    We devise a way to encode a \sygus problem $(\psi(f,\bar{x}), G)$
    over a finite set of examples $E$ as a \emph{reachability problem in a
    recursive program} $\pge{G}{E}$.
    In particular, the program that we construct has an
    assertion that holds if and only the given \sygus problem is
    unrealizable.
    Consequently, \emph{unrealizability} can be proven by establishing
    that the assertion always holds.
    This property can often be established by a conventional
    program-analysis tool.
\end{enumerate}
The encoding mentioned in \itemref{EncodingOfGrammarAndExamples} is
non-trivial for three reasons.
The following list explains each issue, and sketches how
they are addressed

\vspace{1mm}
\mypar{1) Infinitely many terms.}
    We need to model the infinitely many terms generated by the grammar
    of a given synthesis problem $(\psi(f,\bar{x}), G)$.

    To address this issue, we use non-determinism and recursion, and
    give an encoding $\pge{G}{E}$ such that (i) each non-deterministic path $p$ in
    the program $\pge{G}{E}$ corresponds to a possible expression $e_p$ that 
    $G$ can generate, and (ii) for each expression $e$ that $G$  
    can generate, there is a path $p_e$ in $\pge{G}{E}$.
    (There is an isomorphism between paths and the expression-trees of $G$)
    
\mypar{2) Nondeterminism.}
    We need the computation performed along each path $p$ in $\pge{G}{E}$ to mimic the
    execution of expression $e_p$.
    Because the program uses non-determinism, we need to make sure
    that, for a given path $p$ in the program $\pge{G}{E}$, computational steps
    are carried out that mimic the evaluation of $e_p$ for
    \emph{each} of the finitely many example inputs in $E$.

    We address this issue by threading the expression-evaluation
    computations associated with each  example in $E$ through
    the \emph{same} non-deterministic choices.

\mypar{3) Complex Specifications.}
    We need to handle specifications that allow for nested calls of
    the programs being synthesized. 

    For instance, consider the specification $f(f(x)) = x$.
    To handle this specification, we introduce a new 
    variable $y$ and rewrite the specification as
    $f(x) = y \land f(y) = x$.
    Because $y$ is now also used as an input to $f$, we will thread both
    the computations of $x$ and $y$ through the non-deterministic recursive program.


Our work makes the following contributions:
\begin{itemize}
  \item
    We  reduce the \sygus unrealizability problem to a
    reachability problem to which standard program-analysis tools can be
    applied (\sectref{IllustrativeExample} and \sectref{verification}).
  \item 
    We observe that, for many \sygus problems, unrealizability can be
    proven using \emph{finitely many input examples}, and use this idea to
    apply the Counter-Example-Guided Inductive Synthesis (CEGIS) algorithm
    to the problem of proving unrealizability (\sectref{cegis}).
  \item
    We give an encoding of a \sygus problem $(\psi(f,\bar{x}), G)$
    over a finite set of examples $E$ as a reachability problem in a
    nondeterministic recursive program $\pge{G}{E}$, which has the following
    property:
    if a certain assertion in $\pge{G}{E}$ always holds, then the
    synthesis problem is unrealizable (\sectref{verification}).
    \Omit{By threading the expression-evaluation
    computations associated with each of the example inputs in $E$ through
    the \emph{same} non-deterministic choices of $\pge{G}{E}$, our encoding guarantees
    that each path $p$ in the program corresponds to evaluating an
    expression $e_p$ in the grammar $G$ (and vice-versa).}
  \item
    We implement our technique in a tool \name using the ESolver
    synthesizer~\cite{alur2016sygus} as the \sygus solver and the SeaHorn
    tool~\cite{seahorn} for checking reachability.
    \name is able to establish unrealizability for \numSolved out of
    \numBenchmarks variants of benchmarks taken from the \sygus competition.
    	In particular, \name solves all benchmarks with no more than 15 productions in the grammar and requiring no more than 9 input examples for proving unrealizability. 
    Existing \sygus solvers lack the ability to prove that
    these benchmarks are unrealizable, and time out on them.
\end{itemize}
\sectref{RelatedWork} discusses related work.
Some additional technical material, proofs, and full experimental results
are given in
\iffull
\apprefsp{AdditionalMaterial}{proofs}{supp-eval}, respectively.
\Omit{
\fi
\cite{long}.
\iffull
}
\fi



\section{Illustrative Example}
\label{Se:IllustrativeExample}

In this section, we illustrate the main components of our framework
for establishing the unrealizability of a \sygus problem.

Consider the \sygus problem to synthesize a function $f$ that computes the maximum of two
variables $x$ and $y$, denoted by $(\psi_{\texttt{max2}}(f,x,y), G_1)$.
The goal is to create $e_f$---an expression-tree for $f$---where $e_f$
is in the language of the following regular-tree grammar $G_1$:
\[\hspace{-5pt}
\begin{array}{r@{~}r@{~~}ll}
\textrm{Start}	&	::=	& \textrm{Plus}(\textrm{Start}, \textrm{Start}) \hspace{0.34mm}
\mid \textrm{IfThenElse}(\textrm{BExpr}, \textrm{Start}, \textrm{Start})
\mid x \mid y \mid 0 \mid 1\\			
\textrm{BExpr}	&	::=	& \textrm{GreaterThan}(\textrm{Start}, \textrm{Start})
\mid \textrm{Not}(\textrm{BExpr})
\mid \textrm{And}(\textrm{BExpr}, \textrm{BExpr})
\end{array}
\]
and $\forall x,y. \psi_{\texttt{max2}}(\sem{e_f},x,y)$ is valid, where 
$\sem{e_f}$ denotes the meaning of $e_f$, and
\[
\psi_{\texttt{max2}}(f,x,y) := f(x,y)\ge x \wedge f(x,y)\ge y\wedge (f(x,y)=x\vee f(x,y) =y).
\]
\sygus solvers can easily find a solution, such as 
\[
e := \textrm{IfThenElse}(\textrm{GreaterThan}(x,y),x,y).
\]


Although many solvers can now find solutions efficiently to many \sygus problems,
there has been effectively no work on the much harder task of proving that
a given \sygus problem is \emph{unrealizable}---i.e., it does not admit a solution.
For example, consider the \sygus problem $(\psi_{\texttt{max2}}(f,x,y), G_2)$,
where $G_2$ is the more restricted grammar with if-then-else operators and conditions
stripped out:
\[\hspace{-5pt}
\begin{array}{r@{~}r@{~~}ll}
\textrm{Start}	&	::=	& \textrm{Plus}(\textrm{Start}, \textrm{Start}) \hspace{0.34mm}
\mid x \mid y \mid 0 \mid 1
\end{array}
\]
This \sygus problem does \emph{not} have a solution, because no expression
generated by $G_2$ meets the specification.\footnote{Grammar $G_2$ generates all linear functions of
	$x$ and $y$, and hence generates an infinite number of functionally distinct programs;
	however, the maximum function cannot be described by a linear function. 
}
However, to the best of our knowledge, current \sygus solvers cannot
prove that such a \sygus problem is unrealizable.
As an example, we use the problem $(\psi_{\texttt{max2}}(f,x,y), G_2)$
discussed in \sectref{Introduction}, and show how unrealizability
can be proven using four input examples: $(0,0)$, $(0,1)$, $(1,0)$, and $(1,1)$.

\begin{figure}[!t]
	\begin{adjustbox}{width=.75\textwidth,center}
	\begin{lstlisting}[style=CStyle]
int I_0;
void Start(int x_0,int y_0){
	if((*@\nd{}@*)){  // Encodes ``Start ::= Plus(Start, Start)'' (*@\label{Li:PlusBegin0}@*)
		Start(x_0, y_0);   (*@\label{Li:StartRecursiveOne0}@*)
		int tempL_0 = I_0;
		Start(x_0, y_0);   (*@\label{Li:StartRecursiveTwo0}@*)
		int tempR_0 = I_0;
		I_0 = tempL_0 + tempR_0;  (*@\label{Li:PlusEnd0}@*)
	}
	else if((*@\nd{}@*))	I_0 = x_0;  // Encodes ``Start ::= x''  (*@\label{Li:StartX0}@*)
	else if((*@\nd{}@*))	I_0 = y_0;  // Encodes ``Start ::= y''  (*@\label{Li:StartY0}@*)
	else if((*@\nd{}@*))	I_0 = 1;    // Encodes ``Start ::= 1''  (*@\label{Li:StartOne0}@*)
	else         I_0 = 0; (*@\label{Li:StartZero0}@*)   // Encodes ``Start ::= 0''
}
	
bool spec(int x, int y, int f){
	return (f>=x && f>=y && (f==x || f==y))
}
	
void main(){
	int x_0 = 0; int y_0 = 1;   // Input example (0,1)  (*@\label{Li:Initialization0}@*)
	Start(x_0,y_0);               (*@\label{Li:MainCallsStart0}@*)
	assert(!spec(x_0,y_0,I_0));   (*@\label{Li:Assertion0}@*)
}	\end{lstlisting}
	\end{adjustbox}
\caption{Program $\pge{G_2}{E_1}$ created during the course of proving the unrealizability
  of $(\psi_{\texttt{max2}}(f,x,y), G_2)$ using the set of input examples $E_1=\{(0,1)\}$.
}
\label{Fi:cexample0}
\end{figure}

Our method can be seen as a variant of Counter-Example-Guided
Inductive Synthesis (CEGIS), in which the goal is to create a program $P$
in which a certain assertion always holds.
Until such a program is created, each round of the algorithm returns
a counter-example, from which we extract an additional input example for the
original \sygus problem.
On the $i^{\textrm{th}}$ round, the current set of input examples $E_i$ is used,
together with the grammar---in this case $G_2$---and the specification of the
desired behavior---$\psi_{\texttt{max2}}(f,x,y)$, to create a candidate program $\pge{G_2}{E_i}$.
The program $\pge{G_2}{E_i}$ contains an assertion, and a standard program analyzer is used
to check whether the assertion always holds.

Suppose that for the \sygus problem $(\psi_{\texttt{max2}}(f,x,y), G_2)$
we start with just the one example input $(0,1)$---i.e., $E_1=\{(0,1)\}$.
\figref{cexample0} shows the initial program $\pge{G_2}{E_1}$ that our method creates.
The function \texttt{spec} implements the predicate $\psi_{\texttt{max2}}(f,x,y)$.
(All of the programs $\{ \pge{G_2}{E_i} \}$ use the same function \texttt{spec}.)
The initialization statements ``\texttt{int x\_0 = 0; int y\_0 = 1;}''
at \lineref{Initialization0} in procedure \texttt{main} correspond to
the input example $(0,1)$.
The recursive procedure \texttt{Start} encodes the productions of grammar $G_2$.
\texttt{Start} is non-deterministic;
it contains four calls to an external function \nd{}, which returns a
non-deterministically chosen Boolean value.
The calls to \nd{} can be understood as controlling whether or not a
production is selected from $G_2$ during a top-down, left-to-right
generation of an expression-tree:
\lineseqref{PlusBegin0}{PlusEnd0} correspond to
``Start ::= Plus(Start, Start),'' and
\linerefspp{StartX0}{StartY0}{StartOne0}{StartZero0} correspond to
``Start ::= x,'' ``Start ::= y,'' ``Start ::= 1,'' and ``Start ::= 0,'' respectively.
The code in the five cases in the body of \texttt{Start} encodes
the semantics of the respective production of $G_2$;
in particular, the statements that are executed along any execution
path of $\pge{G_2}{E_1}$ implement the \emph{bottom-up evaluation of
some expression-tree that can be generated by $G_2$}.
For instance, consider the path that visits statements in the
following order (for brevity, some statement numbers have been
elided):
\begin{equation}
  \label{Eq:PathPlusXOne0}
  \ref{Li:Initialization0}~~\,
  \ref{Li:MainCallsStart0}~~\,(_{\texttt{Start}}~~\,
    \ref{Li:PlusBegin0}~~\,
    \ref{Li:StartRecursiveOne0}~~\,(_{\texttt{Start}}~~\,
      \ref{Li:StartX0}~~\,
    )_{\texttt{Start}}~~\,
    \ref{Li:StartRecursiveTwo0}~~\,(_{\texttt{Start}}~~\,
      \ref{Li:StartOne0}~~\,
    )_{\texttt{Start}}~~\,
    \ref{Li:PlusEnd0}~~\,
  )_{\texttt{Start}}~~\,
  \ref{Li:Assertion0},
\end{equation}
where $(_{\texttt{Start}}$ and $)_{\texttt{Start}}$ indicate entry to,
and return from, procedure \texttt{Start}, respectively.
Path (\ref{Eq:PathPlusXOne0}) corresponds to the top-down, left-to-right generation
of the expression-tree \texttt{Plus(x,1)}, interleaved with the tree's bottom-up evaluation.

\begin{figure}[!t]
	\begin{adjustbox}{width=.85\textwidth,center}
	\begin{lstlisting}[style=CStyle]
int I_0, I_1, I_2, I_3;
void Start(int x_0,int y_0,...,int x_3,int y_3){
	if((*@\nd{}@*)){  // Encodes ``Start ::= Plus(Start, Start)'' (*@\label{Li:PlusBegin}@*)
		Start(x_0, y_0, x_1, y_1, x_2, y_2, x_3, y_3);    (*@\label{Li:StartRecursiveOne}@*)
		int tempL_0 = I_0; int tempL_1 = I_1;
		int tempL_2 = I_2; int tempL_3 = I_3;
		Start(x_0, y_0, x_1, y_1, x_2, y_2, x_3, y_3);     (*@\label{Li:StartRecursiveTwo}@*)
		int tempR_0 = I_0; int tempR_1 = I_1;
		int tempR_2 = I_2; int tempR_3 = I_3;
		I_0 = tempL_0 + tempR_0; 
		I_1 = tempL_1 + tempR_1;
		I_2 = tempL_2 + tempR_2; 
		I_3 = tempL_3 + tempR_3;}    (*@\label{Li:PlusEnd}@*)
	else if((*@\nd{}@*)) {  // Encodes ``Start ::= x''
		I_0 = x_0; I_1 = x_1; I_2 = x_2; I_3 = x_3;}   (*@\label{Li:StartX}@*)
	else if((*@\nd{}@*)) {  // Encodes ``Start ::= y''
		I_0 = y_0; I_1 = y_1; I_2 = y_2; I_3 = y_3;}   (*@\label{Li:StartY}@*)
	else if((*@\nd{}@*)) {  // Encodes ``Start ::= 1''
		I_0 = 1;   I_1 = 1;   I_2 = 1;   I_3 = 1;}           (*@\label{Li:StartOne}@*)
	else {          // Encodes ``Start ::= 0''  (*@\label{Li:StartZero}@*)  
		I_0 = 0;   I_1 = 0;   I_2 = 0;   I_3 = 0;}
}
	
bool spec(int x, int y, int f){
	return (f>=x && f>=y && (f==x || f==y))
}
	
void main(){
	int x_0 = 0; int y_0 = 1;  // Input example (0,1)     (*@\label{Li:InitializationBegin}@*)
	int x_1 = 0; int y_1 = 0;  // Input example (0,0)
	int x_2 = 1; int y_2 = 1;  // Input example (1,1)
	int x_3 = 1; int y_3 = 0;  // Input example (1,0)     (*@\label{Li:InitializationEnd}@*)
	Start(x_0,y_0,x_1,y_1,x_2,y_2,x_3,y_3);               (*@\label{Li:MainCallsStart}@*)
	assert(   !spec(x_0,y_0,I_0) || !spec(x_1,y_1,I_1)    (*@\label{Li:AssertionBegin}@*)
         || !spec(x_2,y_2,I_2) || !spec(x_3,y_3,I_3));       (*@\label{Li:AssertionEnd}@*)
}	\end{lstlisting}
	\end{adjustbox}
\caption{Program $\pge{G_2}{E_4}$ created during the course of proving the unrealizability
  of $(\psi_{\texttt{max2}}(f,x,y), G_2)$ using the set of input examples $E_4=\{(0,0), (0,1), (1,0), (1,1)\}$.
}
\label{Fi:cexample}
\end{figure}

Note that with path (\ref{Eq:PathPlusXOne0}), when control returns to \texttt{main},
variable \texttt{I\_0} has the value 1, and thus the assertion at
\lineref{Assertion0} fails.

A sound program analyzer will discover that some such path exists in the program,
and will return the sequence of non-deterministic choices required to follow
one such path.
Suppose that the analyzer chooses to report path
(\ref{Eq:PathPlusXOne0}); the sequence of choices would be $t, f, t,
f, f, f, t$, which can be decoded to create the expression-tree
\texttt{Plus(x,1)}.
At this point, we have a candidate definition for $f$: $f = x + 1$.
This formula can be checked using an SMT solver to see whether it
satisfies the behavioral specification $\psi_{\texttt{max2}}(f,x,y)$.
In this case, the SMT solver returns ``false.''
One counter-example that it could return is $(0,0)$.

At this point, program $\pge{G_2}{E_2}$ would be constructed using both of the example inputs
$(0,1)$ and $(0,0)$.
Rather than describe $\pge{G_2}{E_2}$, we will describe the final program constructed, $\pge{G_2}{E_4}$
(see \figref{cexample}).

As can be seen from the comments in the two programs, program $\pge{G_2}{E_4}$ has the same
basic structure as $\pge{G_2}{E_1}$.
\begin{itemize}
  \item
    \texttt{main} begins with initialization statements for
    the four example inputs.
  \item
    \texttt{Start} has five cases that correspond to the
    five productions of $G_2$.
\end{itemize}
The main difference is that because the encoding of $G_2$ in \texttt{Start}
uses non-determinism, we need to make sure that along \emph{each}
path $p$ in $\pge{G_2}{E_4}$, each of the example inputs is used to evaluate the
\emph{same} expression-tree.
We address this issue by threading the expression-evaluation
computations associated with each of the example inputs through
the \emph{same} non-deterministic choices.
That is, each of the five ``production cases'' in \texttt{Start} has four
encodings of the production's semantics---one for each of the four
expression evaluations.
By this means, the statements that are executed along path $p$ perform
\emph{four simultaneous bottom-up evaluations} of the expression-tree
from $G_2$ that corresponds to $p$.

Programs $\pge{G_2}{E_2}$ and $\pge{G_2}{E_3}$ are similar to $\pge{G_2}{E_4}$, but their paths carry
out two and three simultaneous bottom-up evaluations, respectively.
The actions taken during rounds 2 and 3 to generate a new
counter-example---and hence a new example input---are similar to what
was described for round 1.
On round 4, however, the program analyzer will determine that the
assertion on \lineseqref{AssertionBegin}{AssertionEnd} always holds,
which means that there is no path through $\pge{G_2}{E_4}$ for which the
behavioral specification holds for all of the input examples.
This property means that
there is no expression-tree that satisfies the specification---i.e.,
the \sygus problem $(\psi_{\texttt{max2}}(f,x,y), G_2)$ is unrealizable.

Our implementation uses the program-analysis tool \seahorn\ \cite{seahorn}
as the assertion checker.
In the case of $\pge{G_2}{E_4}$, \seahorn takes only 0.5 seconds to establish
that the assertion in $\pge{G_2}{E_4}$ always holds.



\section{\sygus, Realizability, and \cegis}
\label{Se:cegis}
\Omit{
In this section, we recall the definition of a Syntax-Guided Synthesis (\sygus) problem, 
the Counterexample-Guided Inductive Synthesis (\cegis) algorithm, and show how
\cegis can be used to prove unrealizability for  \sygus problems.
\Omit{\begin{changebar}
	\sout{(Proofs can be found in \appref{proofs}.)}
\end{changebar}}
}

\subsection{Background}
\paragraph{Trees and Tree Grammars.}
A \emph{ranked alphabet} is a tuple $(\Sigma,rk_\Sigma)$ where
$\Sigma$ is a finite set of symbols and $rk_\Sigma:\Sigma\to\nat$
associates a rank to each symbol.
For every $m\ge 0$, the set of all symbols in $\Sigma$ with rank $m$
is denoted by $\Sigma^{(m)}$.
In our examples, a ranked alphabet is specified by showing the set
$\Sigma$ and attaching the respective rank to every symbol as a
superscript---e.g., $\Sigma=\{+^{(2)},c^{(0)}\}$.
(For brevity, the superscript is sometimes omitted.)
We use $T_\Sigma$ to denote the set of all (ranked) trees over
$\Sigma$---i.e., $T_\Sigma$ is the smallest set such that
\rone $\Sigma^{(0)} \subseteq T_\Sigma$,
\rtwo if $\sigma^{(k)} \in \Sigma^{(k)}$ and $t_1,\ldots,t_k\in T_\Sigma$, then
$\sigma^{(k)}(t_1,\cdots,t_k)\in T_\Sigma$.
In what follows, we assume a fixed ranked alphabet $(\Sigma,rk_\Sigma)$.

In this paper, we focus on \emph{typed} regular tree grammars, in which
each nonterminal and each symbol is associated with a type.
There is a finite set of types $\{ \type_1, \ldots, \type_k \}$.
Associated with each symbol $\sigma^{(i)} \in \Sigma^{(i)}$, there is a type assignment
$a_{\sigma^{(i)}} = (\tau_0, \tau_1, \ldots, \tau_i)$, where $\tau_0$ is called the \emph{left-hand-side type}
and $\tau_1, \ldots, \tau_i$ are called the \emph{right-hand-side types}.
Tree grammars are similar to word grammars, but generate trees over a ranked alphabet
instead of words. 

\begin{definition}[Regular Tree Grammar]
A \textbf{\emph{typed regular tree grammar}} (RTG) is a tuple $G=(N,\Sigma,S,a,\delta)$,
where $N$ is a finite set of non-terminal symbols of arity 0;
$\Sigma$ is a ranked alphabet;
$S\in N$ is an initial non-terminal;
$a$ is a type assignment that gives types for members of $\Sigma\cup N$;
and $\delta$ is a finite set of productions of the form $A_0 \to \sigma^{(i)}(A_1, \ldots, A_i)$,
where for $1 \leq j \leq i$, each $A_j \in N$ is a non-terminal such that if 
$a(\sigma^{(i)}) = (\tau_0, \tau_1, \ldots, \tau_i)$ then 
$a(A_j) = \tau_j$.
\end{definition}

In a \sygus problem, each variable, such as $x$ and $y$ in the example
RTGs in \sectref{Introduction}, is treated as an arity-$0$
symbol---i.e., $x^{(0)}$ and $y^{(0)}$.

Given a tree $t\in T_{\Sigma\cup N}$,
applying a production $r = A\to\beta$ to $t$ produces the tree $t'$ resulting from replacing
the left-most occurrence of $A$ in $t$ with the right-hand side $\beta$.
A tree $t\in T_\Sigma$ is generated by the grammar $G$---denoted by $t\in
L(G)$---iff it can be obtained by applying a sequence of
productions $r_1\cdots r_n$ to the tree whose root is the initial
non-terminal $S$.

%
%

\paragraph{Syntax-Guided Synthesis.}
A \sygus problem is specified with respect to a background theory $T$---e.g., linear arithmetic---and
the goal  is to synthesize a function $f$ that satisfies two constraints provided by the user. 
The first constraint, $\psi(f,\bar{x})$, describes a \emph{semantic property} that $f$ should satisfy. 
The second constraint limits the \emph{search space} $S$ of $f$, and is
given as a set of expressions specified by an RTG $G$
that defines a subset of all terms in $T$. 

\begin{definition}[\sygus]
A \sygus problem over a background theory $T$ is a pair 
$\sy =(\psi(f,\bar{x}), G)$ where
$G$ is a regular tree grammar that only contains terms in $T$---i.e., $L(G)\subseteq T$---and
$\psi(f,\bar{x})$ is a Boolean formula constraining the semantic behavior of the synthesized program $f$.

A \sygus problem is \textbf{\emph{realizable}} if there exists a expression $e\in L(G)$ such that
$\forall \bar{x}. \psi(\sem{e},\bar{x})$ is true. Otherwise we say that the problem is \textbf{\emph{unrealizable}}.
\end{definition}

\begin{theorem}[Undecidability~\cite{CaulfieldRST15}]
Given a \sygus problem $\sy$, it is undecidable to check whether $\sy$ is realizable.
\end{theorem}

%

\paragraph{Counterexample-Guided Inductive Synthesis}

The Counterexample-Guided Inductive Synthesis (\cegis) algorithm is a popular approach to solving synthesis problems. 
Instead of directly looking for an expression that satisfies the
specification $\varphi$ on \emph{all} possible inputs, the \cegis algorithm
uses a synthesizer $S$ that can find expressions that are correct on a
\emph{finite} set of examples $E$.
If $S$ finds a solution that is correct on all elements of $E$, \cegis
uses a verifier $V$ to check whether the discovered solution is also
correct for all possible inputs to the problem.
If not, a counterexample obtained from $V$ is added to the set of examples, and the
process repeats.
More formally, \cegis starts with an empty set of examples $\cex$ and
repeats the following steps:
\begin{enumerate}
  \item
    \label{It:CallSynthesizer}
    Call the synthesizer $S$ to find an expression $e$ such that 
    $\psi^\cex(\sem{e},\bar{x})\eqdef \forall\bar{x}\in\cex.\psi(\sem{e},\bar{x})$
    holds and go to step \ref{It:CallVerifier};
    return \emph{unrealizable} if no expression exists.
  \item
    \label{It:CallVerifier}
    Call the verifier $V$ to find a model $c$ for the formula $\neg \psi(\sem{e},\bar{x})$,
    and add $c$ to the counterexample set $\cex$;
    return $e$ as a valid solution if no model is found.
\end{enumerate}

Because \sygus problems are only defined over first-order decidable theories, any SMT solver
can be used as the verifier $V$ to check whether the formula $\neg \psi(\sem{e},\bar{x})$ is satisfiable. 
On the other hand, providing a synthesizer $S$ to find solutions such that $\forall \bar{x}\in\cex.\psi(\sem{e},\bar{x})$
holds is a much harder problem because $e$ is a second-order term drawn from an infinite
search space.
In fact, checking whether such an $e$ exists is an undecidable problem~\cite{CaulfieldRST15}.
\Omit{
\begin{changebar}\sout{For this reason, existing \sygus solvers based on \cegis cannot prove unrealizability }
\end{changebar}.}

The main contribution of our paper is a reduction of the unrealizability problem---i.e.,
the problem of proving that there is no expression $e\in L(G)$ such that
$\forall \bar{x}\in\cex.\psi(\sem{e},\bar{x})$ holds---to
an unreachability problem (\sectref{verification}).
This reduction allows us to use existing (un)reachability verifiers to
check whether a \sygus instance is unrealizable.

\subsection{\cegis and Unrealizability}
\label{Se:CEGISAndUnrealizability}

The \cegis algorithm is sound but incomplete for proving unrealizability.
Given a \sygus problem $\sy=(\psi(f,\bar{x}),G)$ and a finite set of inputs $\cex$, we denote 
with $\sy^\cex:=(\psi^\cex(f,\bar{x}),G)$ the  corresponding \sygus problem that only
requires the function $f$ to be correct on the examples in $E$.
\begin{lemma}[Soundness]
\label{lem:soundness-cegis}
	If $\sy^\cex$ is unrealizable then $\sy$ is unrealizable.
\end{lemma}


Even when given a perfect synthesizer $S$---i.e., one that
can solve a problem $sy^\cex$ for every possible set $\cex$---there are 
\sygus problems for which the
\cegis algorithm is not powerful enough to prove unrealizability.

\begin{lemma}[Incompleteness]
\label{lem:incompleteness-cegis}
	There exists an unrealizable \sygus problem $\sy$ such that for every finite set of 
	examples $\cex$ the problem $\sy^\cex$ is realizable.
\end{lemma}
Despite this negative result, we will show that a \cegis algorithm can prove unrealizability for many \sygus instances (\sectref{evaluation}).

%
%
%
%
%



\section{From Unrealizability to Unreachability}
\label{Se:verification}

In this section, we show how a \sygus problem for finitely many
examples can be reduced to a reachability problem in a non-deterministic,
recursive program in an imperative programming language.

\subsection{Reachability Problems}
\label{Se:ReachabilityProblems}

A \emph{program} $P$ takes an initial state $I$ as input and outputs a
final state $O$, i.e., $\sem{P}(I)=O$ where $\sem{\cdot}$ denotes the
semantic function of the programming language.
As illustrated in \sectref{IllustrativeExample}, we allow a program to
contain calls to an external function \nd{}, which returns a non-deterministically
chosen Boolean value.
When program $P$ contains calls to \nd{}, we use $\hat{P}$ to denote
the program that is the same as $P$ except that $\hat{P}$ takes an
additional integer input $\texttt{n}$,
and each call $\nd{}$ is replaced by a call to a local function \nextbit{}
defined as follows:
\[
  \texttt{bool nextbit()\{bool b = n\%2; n=n>>1; return b;\}}.
\]
In other words, the integer parameter $\texttt{n}$ of $\hat{P}[\texttt{n}]$
formalizes all of the non-deterministic choices made by $P$ in calls to $\nd{}$.

For the programs $\pge{G}{E}$ used in our unrealizability algorithm, the only
calls to $\nd{}$ are ones that control whether or not a production is
selected from grammar $G$ during a top-down, left-to-right generation
of an expression-tree.
Given $\texttt{n}$, we can decode it to identify which expression-tree $\texttt{n}$
represents.


\begin{example}\label{Exa:TermDecoding}
Consider again the \sygus problem $(\psi_{\texttt{max2}}(f,x,y), G_2)$ discussed in
\sectref{IllustrativeExample}.
In the discussion of the initial program $\pge{G_2}{E_1}$ (\figref{cexample0}), we hypothesized
that the program analyzer chose to report path (\ref{Eq:PathPlusXOne0}) in $P$,
for which the sequence of non-deterministic choices is $t, f, t, f, f, f, t$.
That sequence means that for $\hat{P}[\texttt{n}]$, the value of $\texttt{n}$ is
$1000101~(\textrm{base}~2)$ (or $69~(\textrm{base}~10)$).
The $1$s, from low-order to high-order position, represent choices
of production instances in a top-down, left-to-right generation
of an expression-tree.
(The $0$s represent rejected possible choices.)
The rightmost $1$ in $\texttt{n}$ corresponds to the choice in
\lineref{PlusBegin0} of ``\texttt{Start ::= Plus(Start, Start)}'';
the $1$ in the third-from-rightmost position corresponds to the choice
in \lineref{StartX0} of ``\texttt{Start ::= x}'' as the left child
of the \texttt{Plus} node;
and the $1$ in the leftmost position corresponds to the choice
in \lineref{StartOne0} of ``\texttt{Start ::= 1}'' as the right child.
By this means, we learn that the behavioral specification
$\psi_{\texttt{max2}}(f,x,y)$ holds for the
example set $E_1 = \{ (0,1) \}$
for $f \mapsto \texttt{Plus(x,1)}$. \qed
\end{example} 

\begin{definition}[Reachability Problem]
  \label{De:reachability}
  Given a program $\hat{P}[\texttt{n}]$, containing assertion
  statements and a non-deterministic integer input $\texttt{n}$, we use
  $\re_P$ to denote the corresponding reachability problem. 
  The reachability problem $\re_P$ is \textbf{satisfiable} if there exists a value $n$
  that, when bound to $\texttt{n}$, falsifies any of the assertions in $\hat{P}[\texttt{n}]$.
  The problem is \textbf{unsatisfiable} otherwise.
\end{definition}

\subsection{Reduction to Reachability}
\label{Se:EncodingAlgorithm}

The main component of our framework is an encoding $\en$ that
given a \sygus  problem  $\sy^\cex=(\psi^\cex(f,{x}), G)$ over  a  set
of examples $\cex = \{c_1,\ldots,c_k\}$, 
outputs a program $\pge{G}{\cex}$
such that 
$\sy^\cex$ is \textbf{realizable} if and only if  $\re_{\en(\sy, E)}$ is \textbf{satisfiable}.
In this section, we define all the components of $\pge{G}{\cex}$, and state
the correctness properties of our reduction.
\Omit{
\begin{changebar}
\sout{	(Proofs are found in \appref{proofs}.)}
\end{changebar}}

\mypar{Remark:} In this section, we assume  that in the specification 
$\psi(f,x)$  every occurrence of $f$ has $x$ as input parameter.
We show how to overcome this restriction in 
\iffull
	\sectref{complex-invocations}.
	\Omit{
\fi
	App. A \cite{long}.
\iffull
}
\fi
In the following, we assume that the input $x$ has type $\tau_I$, where $\tau_I$ could be a complex type---e.g., a tuple type.

\mypar{Program construction.} 
Recall that the grammar $G$ is a tuple $(N,\Sigma,S,a,\delta)$. 
First, for each non-terminal $A\in N$, 
the program $\pge{G}{\cex}$ contains $k$ global variables $\{\gv{1}{A},\dots,\gv{k}{A}\}$ of type $a(A)$ that are
used to express the values resulting from evaluating expressions generated from non-terminal $A$ on the
$k$ examples. 
Second, for each non-terminal $A\in N$, the program $\pge{G}{\cex}$ contains a function 
\[
  \texttt{void funcA(}\tau_I\texttt{ v1}, \dots, \tau_I\texttt{ vk)\{\ } \textit{~bodyA}\texttt{ \}}
\]

We denote by $\delta(A)=\{r_1,\dots,r_m\}$ the set of production rules of the form $A\to \beta$ in $\delta$. 
The body $bodyA$ of $\texttt{funcA}$ has the following structure:
\[
\begin{array}{l}
  \texttt{if(\nd{})}~\{\enprod(r_1)\} \\
  \texttt{else if(\nd{})}~\{\enprod(r_2)\} \\
  \ldots \\
  \texttt{else}~\{\enprod(r_m)\}
\end{array}
\]

The encoding $\enprod(r)$ of a production $r=A_0\to b^{(j)}(A_1,\cdots,A_j)$ is defined as follows ($\tau_i$ denotes the type of the term $A_i$):
\[
\begin{array}{l}
  \texttt{funcA1(v1,\dots,vk);}\\
  \tau_1\ \child{1}{1} = \gv{1}{A1}; \dots; \tau_1\ \child{1}{k} = \gv{k}{Aj};\\
  \ldots \\
  \texttt{funcAj(v1,\dots,vk);}\\
  \tau_j\ \child{j}1 = \gv{1}{A1}; \dots; \tau_j\ \child{j}{k} = \gv{k}{Aj};\\
  \gv{1}{A0} = \en_b^1(\child{1}{1},\dots,\child{1}{k})\\
  \ldots\\
  \gv{k}{A0} = \en_b^k(\child{j}{1},\dots,\child{j}{k})
\end{array}
\]
Note that if $b^{(j)}$ is of arity $0$---i.e., if $j=0$---the construction yields $k$ assignments of the form
$\gv{m}{A0} = \en_b^m()$.	

The function $\en^m_b$ interprets the semantics of $b$ on the $m^{\textrm{th}}$ input example. 
We take Linear Integer Arithmetic as an example to illustrate how $\en^m_b$ works.
\[
\begin{array}{r@{\hspace{1ex}}c@{\hspace{1ex}}l@{\hspace{9ex}}r@{\hspace{1ex}}c@{\hspace{1ex}}l}
  \en_{\textrm{0}^{(0)}}^m&:=& \texttt{0}           &    \en_{\textrm{1}^{(0)}}^m&:=& \texttt{1} \\
  \en_{\textrm{x}^{(0)}}^m&:=& \texttt{vi}          &    \en_{\textrm{Equals}^{(2)}}^m(L,R)&:=& (L \texttt{=} R) \\
  \en_{\textrm{Plus}^{(2)}}^m(L,R)&:=& L\texttt{+}R  &  \en_{\textrm{Minus}^{(2)}}^m(L,R)&:=& L\texttt{-}R \\
  \multicolumn{6}{c}{\en_{\textrm{IfThenElse}^{(3)}}^m(B,L,R)~:=~\texttt{if}(B)\ L\texttt{ else}\ R}
\end{array}
\]
	
We now turn to the correctness of the construction.
First, we formalize the relationship between expression-trees in $L(G)$,
the semantics of $\pge{G}{\cex}$, and the number $\texttt{n}$.
Given an expression-tree $e$, we assume that each node $\node$ in $e$ is annotated with the production that has produced
that node.
Recall that $\delta(A)=\{r_1,\dots,r_m\}$ is the set of productions with head $A$
(where the subscripts are indexes in some arbitrary, but fixed order).
Concretely, for every node $\node$, we assume there is a function
$\produ(\node)=(A,i)$, which associates $\node$ with a pair that indicates
that non-terminal $A$ produced $n$ using the production $r_i$
(i.e., $r_i$ is the $i^{\textrm{th}}$ production whose left-hand-side non-terminal is $A$).

We now define how we can extract a number $\num(e)$ for which the program $\hat{P}[\num(e)]$ will 
exhibit the same semantics as that of the expression-tree $e$.
First, for every node $\node$ in $e$ such that $\produ(\node)=(A,i)$, we define the following number:
\[
\numn(\node)=
\begin{cases}
        1\underbrace{0\cdots0}_{i-1}   & \quad \textrm{if}~i < |\delta(A)| \\
	 \underbrace{0\cdots0}_{i-1}    & \quad \textrm{if}~i = |\delta(A)|.
        \end{cases}
\]
The number $\numn(\node)$ indicates what suffix of the value of
$\texttt{n}$ will cause $\texttt{funcA}$ to trigger the code
corresponding to production $r_i$.
Let $\node_1\cdots \node_m$ be the sequence of nodes visited during a
pre-order traversal of expression-tree $e$.
The number corresponding to $e$, denoted by $\num(e)$, is defined
as the bit-vector $\numn(\node_m)\cdots \numn(\node_1)$.

Finally, we add the entry-point of the program, which calls
the function $\texttt{funcS}$ corresponding to the initial
non-terminal $S$, and contains the assertion that encodes our
reachability problem on all the input examples $\cex = \{c_1,\ldots,c_k\}$.
\[
\begin{array}{l}
  \texttt{void main()}\{\\
  ~~~~\tau_I\texttt{ x1 = }c_1\texttt{;} \cdots\texttt{;}\tau_I\ \texttt{xk = }c_k\texttt{;}\\
  ~~~~\texttt{funcS(}\texttt{x1,}\dots\texttt{,}\texttt{ xk);}\\ 
  ~~~~\texttt{assert}\ \bigvee_{1\leq i\leq k} \neg \psi(f,c_i)[\gv{i}{S}/f(x)] \texttt{; // At least one $c_i$ fails} 
  ~\texttt{\}}
\end{array}
\]


\mypar{Correctness.}
We first need to show that the function $\num(\cdot)$ captures
the correct language of expression-trees.
Given a non-terminal $A$, a value $n$, and input values
$i_1,\ldots,i_k$, we use
$\sem{\texttt{funcA}[n]}(i_1,\dots,i_k)=(o_1,\ldots o_k)$ to
denote the values of the variables $\{\gv{1}{A},\dots,\gv{k}{A}\}$ at
the end of the execution of $\texttt{funcA}[\texttt{n}]$ with
the initial value of $\texttt{n} = n$ and input values $x_1,\ldots,x_k$.
Given a non-terminal $A$, we write $L(G,A)$ to denote the set of terms
that can be derived starting with $A$.

\begin{lemma}
\label{lem:soundness-encoding}
Let $A$ be a non-terminal,  $e\in L(G,A)$ an expression, and $\{i_1,\dots,i_k\}$ an input set. 
Then, $(\sem{e}(i_1),\dots,\sem{e}(i_k))=\sem{\texttt{funcA}[\num(e)]}(i_1,\dots,i_k)$.
\end{lemma}

Each procedure $\texttt{funcA}[\texttt{n}](i_1,\dots,i_k)$ that we
construct has an explicit dependence on variable $\texttt{n}$, where
$\texttt{n}$ controls the non-deterministic choices made by the
$\texttt{funcA}$ and procedures called by $\texttt{funcA}$.
As a consequence, when relating numbers and expression-trees,
there are two additional issues to contend with:
\begin{description}
  \item[Non-termination.]
    \label{It:Nontermination}
    Some numbers can cause $\texttt{funcA}[n]$ to fail to terminate.
    For instance, if the case for ``\texttt{Start ::= Plus(Start, Start)}''
    in program $\pge{G_2}{E_1}$ from \figref{cexample0} were moved from the first
    branch (\lineseqref{PlusBegin0}{PlusEnd0}) to the final else case (\lineref{StartZero0}),
    the number $\texttt{n} = 0 = \ldots0000000~(\textrm{base}~2)$ would cause
    \texttt{Start} to never terminate, due to repeated selections of \texttt{Plus}
    nodes.
    However, note that the only assert statement in the program is placed at
    the end of the main procedure.  
    Now, consider a value of $n$ such that $\re_{\en(\sy, E)}$ is satisfiable.
    \defref{reachability} implies that the flow of control will reach and
    falsify the assertion, which implies that $\texttt{funcA}[{n}]$ terminates.
\footnote{If the \sygus problem deals with the synthesis of programs for a language
      that can express non-terminating programs, that would be an additional
      source of non-termination, different from that discussed in item
      	\textbf{Non-termination}.
      That issue does not arise for LIA \sygus problems. Dealing with the more general kind of non-termination
      is postponed for future work.
    }
  \item[Shared suffixes of sufficient length.]
    \label{It:SharedSuffixes}
    In \exref{TermDecoding}, we showed how for program $\pge{G_2}{E_1}$ (\figref{cexample0})
    the number $\texttt{n} = 1000101~(\textrm{base}~2)$ corresponds to the 
    top-down, left-to-right generation of \texttt{Plus(x,1)}.
    That derivation consumed exactly seven bits;
    thus, any number that, written in $\textrm{base}~2$, shares the suffix $1000101$---e.g.,
    $11010101000101$---will also generate \texttt{Plus(x,1)}.
\end{description}


The issue of shared suffixes is addressed in the following lemma:
\begin{lemma}
\label{lem:completeness-encoding}
For every non-terminal $A$ and number $\texttt{n}$ such that
$\sem{\texttt{funcA}[{n}]}(i_1,\dots,i_k) \neq \bot$ (i.e.,
$\texttt{funcA}$ terminates when the non-deterministic choices are
controlled by $\texttt{n}$), there exists a minimal $\texttt{n}'$ that
is a ($\textrm{base}~2$) suffix of $\texttt{n}$ for which
(i) there is an $e \in L(G)$ such that $\num(e) = {n}'$, and
(ii) for every input $\{i_1,\dots,i_k\}$,
we have $\sem{\texttt{funcA}[{n}]}(i_1,\dots,i_k)=\sem{\texttt{funcA}[{n}']}(i_1,\dots,i_k)$.
\end{lemma}

We are now ready to state the correctness property of our construction.
\begin{theorem}
\label{thm:correctness-encoding}
Given a \sygus  problem $\sy^\cex=(\psi_E(f,x), G)$ over a finite set of examples $\cex$, 
		 the problem $\sy^\cex$ is \textbf{realizable}  iff $ \re_{\en(\sy, E)}$ is \textbf{satisfiable}.
\end{theorem}



\section{Implementation and Evaluation}
\label{Se:evaluation}

\name is a tool that can return two-sided answers to unrealizability
problems of the form $\sy = (\psi, G)$.
When it returns \emph{\textbf{unrealizable}}, no expression-tree in $L(G)$ satisfies $\psi$;
when it returns \emph{\textbf{realizable}}, some $e \in L(G)$ satisfies $\psi$;
\name can also time out.
\name incorporates several existing pieces of software.
\begin{enumerate}
  \item
    The (un)reachability verifier \seahorn is applied to the
    reachability problems of the form $\re_{\en(\sy,\cex)}$ created during
    successive \cegis rounds.
  \item
    The SMT solver \z3 is used to check whether a generated expression-tree $e$
    satisfies $\psi$.
    If it does, \name returns \emph{\textbf{realizable}} (along with $e$);
    if it does not, \name creates a new input example to add to $\cex$.
\end{enumerate}
It is important to observe that \seahorn, like most reachability
verifiers, is only sound for \emph{\textbf{un}}satisfiability---i.e., if
\seahorn returns \emph{\textbf{unsatisfiable}}, the reachability
problem is indeed unsatisfiable.
Fortunately, \seahorn's one-sided answers are in the correct direction
for our application:
to prove unrealizability, \name only requires the reachability
verifier to be sound for unsatisfiability.

There is one aspect of \name that differs from the technique that has been
presented earlier in the paper.
While \seahorn is sound for \emph{\textbf{un}}reachability, it is not sound for
reachability---i.e., it cannot soundly prove whether a synthesis problem is realizable.
To address this problem, to check whether a given \sygus problem $\sy^\cex$ is realizable on the finite set
of examples $\cex$, \name also calls the \sygus solver ESolver
\cite{alur2016sygus} to synthesize an expression-tree $e$ that satisfies
$\sy^\cex$.\footnote{We chose ESolver because on the benchmarks we considered, ESolver outperformed
  other \sygus solvers (e.g., CVC4~\cite{cvc4}).
}

In practice, for every intermediate problem $\sy^\cex$ generated by
the \cegis algorithm, \name runs the ESolver on $sy^\cex$ and \seahorn
on $\re_{\en(\sy,\cex)}$ in \emph{parallel}.
If ESolver returns a solution $e$, \seahorn is interrupted, and Z3 is
used to check whether $e$ satisfies $\psi$.
Depending on the outcome, \name either returns \emph{\textbf{realizable}}
or obtains an additional input example to add to $\cex$.
If \seahorn returns \emph{\textbf{unsatisfiable}}, \name returns \emph{\textbf{unrealizable}}.

Modulo bugs in its constituent components, \name is sound for both
realizability and unrealizability, but because of 
Lemma~\ref{lem:incompleteness-cegis} and the incompleteness of \seahorn, \name is not complete for unrealizability.

%

\smallskip
\noindent\textbf{Benchmarks.}
We perform our evaluation on \numBenchmarks variants of the 
60 LIA benchmarks from the LIA \sygus competition track~\cite{alur2016sygus}.
We do not consider the other \sygus benchmark track, Bit-Vectors, because
the \seahorn verifier is unsound for most bit-vector operations--e.g., bit-shifting.
We used three suites of benchmarks.
\textsc{LimitedIf} (resp.\ \textsc{LimitedPlus}) contains \numIfBenchmarks
(resp.\ \numPlusBenchmarks) benchmarks in which the grammar bounds the number
of times  an \textrm{IfThenElse} (resp.\ \textrm{Plus}) operator can appear
in an expression-tree to be $1$ less than the number
required to solve the original synthesis problem. 
We used the tool
  \textsc{Quasi} to  automatically generate the restricted grammars.
\textsc{LimitedConst} contains \numConstBenchmarks benchmarks in which the grammar
allows the program to contain only constants that are coprime to any constants
that may appear in a valid solution---e.g., the solution requires using odd numbers,
but the grammar only contains the constant $2$.
The numbers of benchmarks in the three suites differ because for certain benchmarks
it did not make sense to create a limited variant---e.g., if the smallest program
consistent with the specification contains no \textrm{IfThenElse} operators,
no variant is created for the \textsc{LimitedIf} benchmark.
In all our benchmarks, the grammars describing the search space contain infinitely many terms.

Our experiments were performed on an Intel Core i7 4.00GHz CPU, with 32GB of RAM,
running Lubuntu 18.10 via VirtualBox.
We used version 4.8 of Z3,
commit 97f2334 of \seahorn, and
commit d37c50e of ESolver.
The timeout for each individual \seahorn/ESolver call
is set at 10 minutes.




\Omit{
The detailed results of running \name on our benchmarks are shown in
\appref{supp-eval}. }

\smallskip
\noindent\textbf{Experimental Questions.}
Our experiments were designed to answer the questions posed below.
\begin{resq}
	Can \name prove unrealizability for variants of real \sygus
	benchmarks, and how long does it take to do so? 
\end{resq}

\finding{Finding: \name can prove unrealizability for $\numSolved/\numBenchmarks$ benchmarks.}
For the $\numSolved$ benchmarks solved by \name, the average time taken is $\avgtime$s.
The time taken to perform the last iteration of the algorithm---i.e., the time taken by \seahorn to
return \textbf{unsatisfiable}---accounts for $87\%$ of the total running time.

\name can solve all of the \textsc{LimitedIf} benchmarks for
which the grammar allows at most one \textrm{IfThenElse} operator.
Allowing more \textrm{IfThenElse} operators in the grammar leads to
larger programs and larger sets of examples, and consequently the
resulting reachability problems are harder to solve for \seahorn.

For a similar reason, \name can solve only one of the \textsc{LimitedPlus} benchmarks. 
All other \textsc{LimitedPlus} benchmarks allow $5$ or more
\textrm{Plus} statements, resulting in grammars that have at least $130$
productions.

\name can solve all \textsc{LimitedConst} benchmarks because these 
require few examples and result in small encoded programs.

%

%
%

\begin{resq}
	How many examples does \name use to prove unrealizability and
	how does the number of examples affect the performance of \name? 
\end{resq}
\vspace{-2mm}\noindent
Note that Z3 can produce different models for the same query, and thus different runs of NOPE can produce different sequences of example. Hence, there is no guarantee that NOPE finds a good sequence of examples that prove unrealizability.
One measure of success is whether \name is generally able to find
a small number of examples, when it succeeds in proving unrealizability.

\finding{Finding: Nope used 1 to 9 examples to prove unrealizability
for the benchmarks on which it terminated.}
Problems requiring large numbers of examples could not be solved because
either ESolver or \seahorn timeouts---e.g.,
on the problem max4, \name gets to the point where the \cegis loop has
generated 17 examples, at which point ESolver exceeds the timeout threshold.

\pgfplotsset{every axis/.append 
style={	
label style={font=\scriptsize}, 
tick label style={font=\scriptsize}}
}
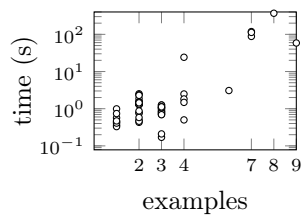
\begin{wrapfigure}{r}{.35\textwidth}
\centering
\vspace{-7mm}
\hspace{-1mm}
\begin{tikzpicture}[scale=1,every mark/.append style={mark size=1pt}]
\begin{axis}[%
ymode=log,
ylabel absolute, ylabel style={yshift=-3mm},
xlabel absolute, xlabel style={yshift=1mm},
xtick =data,
xtick={2,3,4,7,8,9},
width=\linewidth,
height=0.8\linewidth,
scatter/classes={%
c={mark=*,draw=black,
	mark options={solid, mark size=2pt},
	style={solid, fill=white}, mark size =1.2pt},
d={mark=diamond*,draw=black,
	mark options={dashed},
	style={solid, fill=white},mark size = 1.5pt},
a={mark=*,draw=black,
	mark options={solid, mark size=2pt},
	style={solid, fill=white}, mark size =1.2pt},
b={mark=*,draw=black,
	mark options={solid, mark size=2pt},
	style={solid, fill=white}, mark size =1.2pt}},
xmin = 0,
xmax = 9,
ymin = 0,
ymax = 400,
xlabel={examples},
ylabel={time (s)}]
\addplot[scatter,only marks,%
scatter src=explicit symbolic]%
table[meta=label] {
x y label
2	0.78	c
3	1.26	c
3	1.25	c
3	1.01	c
3	0.87	c
3	0.85	c
3	0.97	c
3	0.7	c
3	0.8	c
3	1.09	c
3	1.13	c
3	0.73	c
3	0.77	c
3	1.06	c
2	1.3	c
2	1.46	c
2	1.31	c
2	1.28	c
2	2.52	c
2	1.35	c
2	1.41	c
2	1.43	c
2	2.37	c
2	1.56	c
2	0.76	c
2	1.87	c
2	1.33	c
2	1.53	c
2	1.5	c
2	1.44	c
2	2.29	c
2	0.87	c
1	0.36	c
4	0.5	c
1	0.57	c
1	0.44	c
1	0.99	c
6	3.08	c
4	2.49	c
4	1.83	c
4	24.18	c
1	0.33	c
1	0.41	c
1	0.47	c
1	0.74	c
3	0.69	b
4	1.48	a
9	58.57	a
3	0.17	a
3	0.21	a
7	87.92	a
7	118.77	a
7	112.78	a
3	1.12	a
2	0.43	a
2	0.49	a
2	0.46	a
2	0.58	a
8	369.57	a
};
\end{axis}
\end{tikzpicture}
\caption{Time vs examples.}
\vspace{-9mm}
\label{fig:extime}
\end{wrapfigure}
\finding{Finding: The number of examples required to prove unrealizability
depends mainly on the arity of the synthesized function and the complexity of the grammar.}
The number of examples seems to grow quadratically with the number of bounded operators allowed in the grammar. 
In particular, problems in which the grammar allows zero
\textrm{IfThenElse} operators require 2--4 examples, while
problems in which the grammar allows one \textrm{IfThenElse}
operator require 7--9 examples.


Figure~\ref{fig:extime} plots the running time of \name against the
number of examples generated by the \cegis algorithm.
\finding{Finding: The solving time appears to grow exponentially with the number of examples} required to prove unrealizability.




\section{Related Work}
\label{Se:RelatedWork}

The \sygus formalism was introduced  as a unifying framework to
express several synthesis problems~\cite{sygus}.
Caulfield et al.~\cite{CaulfieldRST15} proved that it is undecidable to determine whether
a given \sygus problem is realizable.
Despite this negative result, there are  several \sygus solvers that 
compete in yearly \sygus competitions~\cite{alur2016sygus} and can
efficiently produce solutions to \sygus problems when a solution exists.
Existing \sygus synthesizers fall into three categories: 
\rone Enumeration solvers enumerate programs with respect to a given total order~\cite{esolver}. 
If the given problem is unrealizable, these solvers 
typically only terminate if the 
language of the grammar is finite or contains finitely many functionally distinct programs.
While in principle certain enumeration solvers can prune infinite portions of the search space,
none of these solvers could prove unrealizability for any of the benchmarks considered in this paper.
\rtwo Symbolic solvers  reduce the synthesis problem to a constraint-solving problem~\cite{cvc4}. 
These solvers cannot reason about grammars that restrict allowed terms, and resort to enumeration whenever the candidate solution produced by the constraint solver is not in the restricted search space.
Hence, they also cannot prove unrealizability.
\rthree Probabilistic synthesizers randomly search the search space,
and are typically unpredictable~\cite{Schkufza0A16}, providing no guarantees in terms of unrealizability. 
\Omit{
\begin{changebar}
\sout{To the best of our knowledge, this paper presents the first technique
for proving that a \sygus problem is unrealizable when no solution exists. 
}\end{changebar}}


\mypar{Synthesis as Reachability.}
CETI~\cite{Nguyen18} introduces a technique for encoding
template-based synthesis problems as reachability problems.
The CETI encoding only applies to the specific setting in which
\rone the search space is described by an imperative program with a
\emph{finite number} of holes---i.e., the values that the synthesizer
has to discover---and \rtwo the specification is given
as a finite number of input-output test cases with which the target program should agree.
Because the number of holes is finite, and all holes correspond to
values (and not terms), the reduction to a reachability problem only
involves making the holes global variables in the program (and no more
elaborate transformations).

In contrast, our reduction technique handles search spaces that are
described by a grammar, which in general consist of an infinite set of
terms (not just values).
Due to this added complexity, our encoding has to account for
(i) the semantics of the productions in the grammar, and
(ii) the use of non-determinism to encode the choice of grammar productions.
Our encoding creates one expression-evaluation computation for each of
the example inputs, and threads these computations through the program
so that each expression-evaluation computation makes use of the
\emph{same} set of non-deterministic choices.

Using the input-threading, our technique can handle
specifications that contain nested calls of the synthesized program
 (e.g., $f(f(x)) = x$). ( 
\iffull
\sectref{complex-invocations}.)
\Omit{
	\fi
App. A \cite{long}.)
\iffull
}
\fi

The input-threading technique builds a \emph{product program} that perform multiple executions of the same function as done in relational program verification \cite{barthe2011relational}. Alternatively, a different encoding could use 
multiple function invocations on individual inputs and require the verifier to thread the  same bit-stream for all input  evaluations. In general, verifiers perform much better on product programs~\cite{barthe2011relational}, which motivates
our choice of encoding.

\mypar{Unrealizability in Program Synthesis.}
For certain synthesis problems---e.g., reactive
synthesis~\cite{Bloem15}---the realizability problem is decidable.
The framework tackled in this paper, \sygus, is orthogonal to such problems,
and it is undecidable to check whether a given \sygus problem is realizable~\cite{CaulfieldRST15}.

Mechtaev et al.~\cite{Mechtaev18} propose to use a variant of \sygus to efficiently prune irrelevant paths
 in a symbolic-execution engine.
In their approach, for each path $\pi$ in the program, a synthesis
problem $p_\pi$ is generated so that if $p_\pi$ is unrealizable, the
path $\pi$ is infeasible.
The synthesis problems generated by Mechtaev et al.\ (which are not directly expressible
in \sygus) are decidable
because the search space is defined by a finite set of templates, and the
synthesis problem can be encoded by an SMT formula.
To the best of our knowledge, our technique is the first one that can
check unrealizability of general \sygus problems in which the search space is an \emph{infinite set of functionally distinct  terms}. 

\noindent\textbf{Acknowledgment}	
This work was supported, in part, by a gift from Rajiv and Ritu Batra; by AFRL under DARPA MUSE award FA8750-14-2-0270 and DARPA STAC award FA8750-15-C-0082; by ONR under grant N00014-17-1-2889; by NSF under grants CNS-1763871 and CCF-1704117; and by the UW-Madison OVRGE with funding from WARF. 
\Omit{
\loris{is the next sentence necessary? I've never seen it in any paper}
Any opinions, findings, and conclusions or recommendations expressed in this publication are those of the authors, and do not necessarily reflect the views of the sponsoring agencies.
}


%
%
%
%
%
%
	\bibliographystyle{abbrv}
	\bibliography{ref}
\iffull
\appendix
\setcounter{theorem}{1}
\setcounter{lemma}{0}

\section{Additional Material}
\label{App:AdditionalMaterial}

\subsection{Encoding in the Presence of Nested Function-Invocations}
\label{Se:complex-invocations}

In \sectref{EncodingAlgorithm}, we presented a simplified encoding that relied on the 
specification $\psi(f,x)$ to only involve function invocations of the form $f(x)$,
where $x$ represents the input parameter of the function to be synthesized.
In this section, we show with a simple example how such a restriction can be overcome.

Consider the following semantic specification that involves multiple
invocations of the function $f$ on different arguments, as well as
nested function calls:
\[
	\psi_1(f,x)\eqdef f(f(x))=f(x+x).
\]
By introducing new input variables and performing the proper
refactoring, we can rewrite $\psi_1$ as the following specification,
where $f$ is always called on a single input variable:
\[
  \psi_2(f,x,y_1,y_2,y_3,y_4)
  \eqdef
  \left[\begin{array}{cl}
           & f(x) = y_1 \wedge f(y_1) = y_2 \\
    \wedge & x+x = y_3 \wedge f(y_3) = y_4
  \end{array}\right] \rightarrow y_2 = y_4.
\]
It is now easy to adapt our encoding to operate over this new specification.
First, the program $\pge{G}{E}$ will now operate over input examples of the
form $c = \{w_1,\ldots, w_k\}$, where each  example $c$ is a tuple corresponding
to the values of variables $\{x,y_1,y_2,y_3,y_4\}$.
Second, the program will need to compute  the values of all possible
calls of $f$ on the various input parameters.
Hence, for every expression $f(z)$ in $\psi_2$, 
non-terminal $A$, and  example $w_i$, the program
$\pge{G}{E}$ will have a global variable $\gva{z}{i}{A}$ 
computing the value of the expression generated by $A$ parametrized by $z$,
with respect to the values in input example $w_i$.

For instance, assume that the input grammar has a production $A\to \pi_1$ 
that generates an access on the first parameter of the function to be
synthesized, and assume that we currently only have one input example.
The corresponding code for the production would be
\[
\begin{array}{l}
  \texttt{funcA\ (int\ v\_x, int\ v\_{y1}, int\ v\_{y2},  int\ v\_{y3}, int\ v\_{y4}) \{}\\
  ~~~~\texttt{if(\nd{}) \{}\\
  ~~~~~~~~\gva{x}{1}{A} \texttt{ = v\_{x};} ~~~ 	\texttt{ // Computing  $f(x)$} \\
  ~~~~~~~~\gva{y1}{1}{A} \texttt{ = v\_{y1};}	\texttt{ // Computing  $f(y_1)$} \\
  ~~~~~~~~\gva{y3}{1}{A} \texttt{ = v\_{y3};} 	\texttt{ // Computing  $f(y_3)$}	\\	
  ~~~~\texttt{\}}\\
  ~~~~\texttt{...}
\end{array}
\]
In summary, thanks to the ability to execute a finite number of inputs
in lock-step, our encoding can handle specifications that contain nested
function-invocations.

\subsection{Overcoming a Quirk of \seahorn}
\label{Se:OvercomingAQuirkofSeaHorn}

Because \seahorn is unsound for satisfiability, it can report
that some expression-tree satisfies behavioral specification $\psi$,
when in fact no such expression-tree exists.
In effect, \seahorn overapproximates the set of reachable states, and
erroneously concludes that the assertion in $\re_{\en(\sy,\cex)}$
can be falsified (i.e., all example inputs satisfy $\psi$).
We encountered this situation in our experiments;
for some unknown reason, when the following two productions
were included in the grammar, \seahorn would report that
$\re_{\en(\sy,\cex)}$ was \emph{\textbf{satisfiable}} in cases
when it should have reported \emph{\textbf{unsatisfiable}}:
\begin{equation}
  \label{Eq:BExprProductions}
  \textrm{BExpr} ::= \textrm{Not}(\textrm{BExpr}) \mid \textrm{And}(\textrm{BExpr}, \textrm{BExpr})
\end{equation}

For the examples on which this happened, we found that we could
delete these two productions, which resulted in a grammar of
equivalent expressiveness.
That is, because the grammar still contained the \textrm{IfThenElse}
operator, for all expressions $e_1$, $e_2$, $e_3$, and $e_4$, the
expression $\textrm{IfThenElse}(\textrm{Not}(e_1),e_2,e_3)$
is equivalent to $\textrm{IfThenElse}(e_1,e_3,e_2)$,
and the expression $\textrm{IfThenElse}(\textrm{And}(e_1,e_2),e_3,e_4)$
is equivalent to $\textrm{IfThenElse}(e_1,\textrm{IfThenElse}(e_2,e_3,e_4),e_4)$.
When we ran the same \sygus problem $\sy$ with productions
(\ref{Eq:BExprProductions}) removed from the grammar,
\seahorn reported that $\re_{\en(\sy,\cex)}$ was unsatisfiable.
Because \seahorn is sound for unsatisfiability, the latter
is the correct answer, and demonstrates that \sygus problem
$\sy$ (in both modified and unmodified form) is unrealizable.

Because the expressibility of the grammar is unchanged with and
without productions (\ref{Eq:BExprProductions}), these examples
demonstrate that the effect is caused by some overapproximation made
by \seahorn, triggered by productions (\ref{Eq:BExprProductions}) and
the encoding described in \sectref{EncodingAlgorithm}.



\section{Proofs of Theorems}
\label{App:proofs}

\begin{lemma}[Soundness]
	If $\sy^\cex$ is unrealizable then $\sy$ is unrealizable.
\end{lemma}
\begin{proof}
For every expression $e$ and input $\bar{c}$ we have that $\psi(\sem{e},\bar{c})\Rightarrow \psi_{\cex}(\sem{e},\bar{c})$ and by contraposition $\neg\psi^\cex(\sem{e},\bar{c})\Rightarrow\neg \psi(\sem{e},\bar{c})$. 
Hence, the lemma holds.\qed
\end{proof}

\begin{lemma}[Incompleteness]
	There exists an unrealizable \sygus problem $\sy$ such that for every finite set of 
	examples $\cex$ the problem $\sy^\cex$ is realizable.
\end{lemma}
\begin{proof}
Let $sy_{\eq}=(\psi_{\eq}(f,x),G_{\eq})$ be the \sygus problem over the theory of linear-integer arithmetic such that
$\psi_{\eq}(f,x)\eqdef f(x)=x$;
that is, $\psi_{\eq}(f,x)$ is a predicate denoting that $f$ should implement the identity function.
Let $G_{\eq}$ be the following grammar:
\[
	\hspace{-5pt}
	\begin{array}{r@{~}r@{~~}ll}
	\textrm{Start}	&	::=	& \textrm{Plus}(\textrm{Start}, \textrm{Start}) \hspace{0.34mm}
                                             \mid \textrm{IfThenElse}(\textrm{BExpr}, \textrm{Start}, \textrm{Start})
                                             \mid 0 \mid 1\\			
	\textrm{BExpr}	&	::=	& \textrm{Equals}(x, \textrm{Start})                                             
	\end{array}
	\]
	The problem $sy_{\eq}$ is unrealizable. Because the grammar does not contain the variable $x$,
	every expression $e=L(G_{\eq})$ can only produce a finite number of constant outputs.
	However, for every set of examples $\cex=\{n_1,\ldots,n_k\}$ the following expression $e_\cex\in L(G_{\eq})$
	is a valid solution to $\sy^\cex$ (i.e., $\psi^\cex_{\eq}(\sem{e^\cex},x)$ holds):
	\[
	\begin{array}{l}
	  \textrm{IfThenElse}(\textrm{Equals}(x,T(n_1)), T(n_1),\\
	  ~~~~\textrm{IfThenElse}(\textrm{Equals}(x,T(n_2)), T(n_2), \ldots, T(n_k)\ldots)
	\end{array}
       \]
where $T(n)$ is the expression-tree corresponding to $0+\underbrace{1+\ldots+1}_{n}$.
Hence, the \cegis algorithm will never terminate for $\sy_{\eq}$.\qed
\end{proof}

\begin{lemma}
For every non-terminal $A$,  expression $e\in L(G,A)$, and input set $\{i_1,\dots,i_k\}$, 
\[
  (\sem{e}(i_1),\dots,\sem{e}(i_k))=\sem{\texttt{funcA}[\num(e)]}(i_1,\dots,i_k)
\]
\end{lemma}	
\begin{proof}
The proof is by structural induction on $e$.
Let $\node$ denote the root of $e$, and
$A \to \sigma^{(j)}(A_1, \ldots, A_j)$ denote the production instance at $\node$.
Note that $\num(e)=\num(e_j)\cdots \num(e_1) \numn(\node)$.

Suppose that $e=\node$ is a leaf node;
that is, the tree at $\node$ is an instance of a production of the
form $A \to \node^{(0)}$.
Because $\num(e) = \numn(\node^{(0)})$, for every input set $\{i_1,\dots,i_k\}$,
$\texttt{funcA}[\num(e)]$ selects the branch in $\texttt{funcA}$
that captures the semantics of $A \to \node^{(0)}$.
In that code, $e$ is evaluated on the $k$ values $\{i_1,\dots,i_k\}$.
Therefore,
$(\sem{e}(i_1),\dots,\sem{e}(i_k))=\sem{\texttt{funcA}[\num(e)]}(i_1,\dots,i_k)$
holds.

{\bf Inductive step}:
Let $e = \sigma^{(j)}(e_1, \ldots, e_j)$, where the property to be shown
is assumed to hold for each of the $e_l$.
For each $e_l$, let $\node_l$ be the root of $e_l$.

The procedure $\texttt{funcA}[\num(e)]$ uses $\numn(\node)$ to select
the branch $B$ in $\texttt{funcA}$ that captures the semantics of the
production $A \to \sigma^{(j)}(A_1, \ldots, A_j)$.
For every input set $\{i_1,\dots,i_k\}$, the induction hypothesis ensures
that the following property holds:
for $1 \leq l \leq j$,
$(\sem{e_l}(i_1),\dots,\sem{e_l}(i_k))=\sem{\texttt{funcAl}[\num(e_l)]}(i_1,\dots,i_k)$.
Therefore, each call to a procedure $\texttt{funcAl}$ in $B$ computes
the $k$ intermediate answers that correspond to the evaluation of $e_l$ on the
$k$ values $\{i_1,\dots,i_k\}$.
The code in $B$ that follows the final call to $\texttt{funcAj}$ uses the
collections of intermediate results to finish $k$ computations of the semantics
of $A \to \sigma^{(j)}(A_1, \ldots, A_j)$.
Therefore, $(\sem{e}(i_1),\dots,\sem{e}(i_k))=\sem{\texttt{funcA}[\num(e)]}(i_1,\dots,i_k)$
holds.\qed
\end{proof}

\begin{lemma}
For every non-terminal $A$ and number $\texttt{n}$ such that
$\sem{\texttt{funcA}[{n}]}(i_1,\dots,i_k) \neq \bot$ (i.e.,
$\texttt{funcA}$ terminates when the non-deterministic choices are
controlled by $\texttt{n}$), there exists a minimal $\texttt{n}'$ that
is a ($\textrm{base}~2$) suffix of $\texttt{n}$ for which
(i) there is an $e \in L(G)$ such that $\num(e) = {n}'$, and
(ii) for every input $\{i_1,\dots,i_k\}$,
\[
  \sem{\texttt{funcA}[{n}]}(i_1,\dots,i_k)=\sem{\texttt{funcA}[{n}']}(i_1,\dots,i_k).
\]
\end{lemma}	
\begin{proof}
Assume that the computation $\sem{\texttt{funcA}[{n}]}(i_1,\dots,i_k)$ terminates.
Let $b_1, \ldots, b_j$ be the finite sequence of bits drawn by $\nd{}$
throughout the computation.

\noindent
{\it Proof of (i)}: Let $e$ be the expression-tree generated top-down,
left-to-right using the sequence $b_1, \ldots, b_j$.
Let $\texttt{n}'$ be the binary number $b_j \cdots b_1$.
Because $\num(e)$ is the concatenation, in right-to-left order,
of the sequence of $\num(\cdot)$ values for the nodes of $e$
visited during a pre-order traversal, $\num(e) = \texttt{n}'$.

\noindent
{\it Proof of (ii)}:
Property (ii) holds because $\texttt{n}$ and $\texttt{n}'$ agree on the
($\textrm{base}~2$) suffix $b_j \cdots b_1$, and exactly $j$ bits are
used during the executions of both $\texttt{funcA}[{n}](i_1,\dots,i_k)$
and $\texttt{funcA}[{n}'](i_1,\dots,i_k)$---which also shows that
$\texttt{n}' = b_j \cdots b_1~(\textrm{base}~2)$ is minimal.\qed
\end{proof}

\begin{theorem}
Given a \sygus  problem $\sy^\cex=(\psi_E(f,x), G)$ over a finite set of examples $\cex$, 
\[
		 \sy^\cex\text{ is \textbf{realizable} } \Longleftrightarrow  \re_{\en(\sy, E)}\text{ is \textbf{satisfiable}}
\]
\end{theorem}
\begin{proof}
\mypar{$\Rightarrow$ direction}:
Assume that $\sy^\cex$  is \textbf{realizable}. Then there exists an expression $e\in L(G)=L(G,S)$ such that
$\forall x\in E.\psi(\sem{e},x)$.
By Lemma~\ref{lem:soundness-encoding}, for every $\{i_1,\dots,i_k\}$,
$(\sem{e}(i_1),\dots,\sem{e}(i_k))=\sem{\texttt{funcA}[\num(e)]}(i_1,\dots,i_k)$.
Hence, the assertion in program $\en(\sy, E)$ is false and
the reachability problem $\re_{\en(\sy, E)}$ is \textbf{satisfiable}.

\mypar{$\Leftarrow$ direction}:
Assume that $\re_{\en(\sy, E)}$ is \textbf{satisfiable}. Then there exists a value of $n$ 
that makes the assertion in program $\en(\sy, E)$ false (i.e., the specification
holds for all inputs $c_i \in \cex$).
By Lemma~\ref{lem:completeness-encoding}, there exists a minimal $n'$ for
which the program has equivalent semantics (in particular, the assertion
in $\en(\sy, E)$ is still false),
and there exists an expression $e\in L(G)$ such that $\num(e)=n'$.
Hence, $e$ is a solution to \sygus problem $\sy^\cex$;
i.e., $\sy^\cex$ is \textbf{realizable}.\qed
\end{proof}



\section{Supplementary Evaluation Results}
\label{App:supp-eval}

The complete results of our evaluation are shown in
Tables~\ref{Ta:results} and~\ref{Ta:results2}.
For brevity, in Table~\ref{Ta:results} we omit consecutive benchmarks
on which \name times out---e.g., the ``$\dots$'' between benchmarks
max4 and max15 represents 10 benchmarks from max5 to max14 for which \name times out.

The tables present the \textbf{number of nonterminals} and the
\textbf{number of productions} in the grammar of each benchmark, the
\textbf{number of examples} used to prove unrealizability, the
total time taken by \name, and the time taken by \seahorn for the
last (un)reachability problem.
For benchmarks on which \name times out, the value given for
``\textbf{number of examples}''
is the number of examples generated by the CEGIS loop when
\name times out.

\begin{table*}[!t]
\caption{
Performance of \name on \textsc{LimitedIf} and \textsc{LimitedPlus} benchmarks.
\xmark\ denotes a timeout.
}
\scriptsize
\centering
\begin{tabular}{cc|rrrrr}
& \multirow{2}{*}{\bf Problem} & \multirow{1}{*}{\bf number of} & \multirow{1}{*}{\bf number of} & \multirow{1}{*}{\bf number of} & \multirow{1}{*}{\bf total} & \multirow{1}{*}{\bf \seahorn} \\
&    &  {\bf nonterminals}  &  {\bf productions}  &  {\bf examples}  &  \multirow{1}{*}{\bf time (s)}  & \multirow{1}{*}{\bf time (s)}   \\
\cline{1-7}

\parbox[t]{3mm}{\multirow{29}{*}{\rotatebox[origin=c]{90}{\!\!\!\!\!\!\textsc{LimitedIf}}}}
&	max2				&	1	&	5	&		4	&	1.48	&	0.53	\\
&	max3				&	3	&	15	&		9	&	58.57	&	50.21	\\
&	max4				&	5	&	34	&		17	&	\xmark	&	\xmark	\\
&   $\dots$				&		&		&			&	\xmark	&	\xmark	\\
&   max15				&	27	&	348		&		1	&	\xmark	&	\xmark	\\
&	array\_sum\_2\_5	&	1	&	5	&		3	&	0.69	&	0.17	\\
&	array\_sum\_2\_15	&	1	&	5	&		3	&	0.87	&	0.21	\\
&	array\_sum\_3\_5	&	3	&	15	&		7	&	101.44	&	87.92	\\
&	array\_sum\_3\_15	&	3	&	15	&		7	&	134.87	&	118.77	\\
&	array\_sum\_4\_5	&	5	&	34	&		14	&	\xmark	&	\xmark	\\
&	array\_sum\_4\_15	&	5	&	34	&		16	&	\xmark	&	\xmark	\\
&   $\dots$				&		&		&			&	\xmark	&	\xmark	\\
&   array\_sum\_10\_5	&	19	&	149		&		1	&	\xmark	&	\xmark	\\
&   array\_sum\_10\_15	&	19	&	149		&		1	&	\xmark	&	\xmark	\\
&	array\_search\_2	&	3	&	15	&		7	&	112.78	&	87.32	\\
&	array\_search\_3	&	5	&	34	&		17	&	\xmark	&	\xmark	\\
&   $\dots$				&		&		&			&	\xmark	&	\xmark	\\
&	array\_search\_15	&	27	&	348		&		1	&	\xmark	&	\xmark	\\
&	mpg\_example1		&	1	&	7	&		3	&	1.12	&	0.38	\\
&	mpg\_example2		&	9	&	60	&		17	&	\xmark	&	\xmark	\\
&	mpg\_example3		&	5	&	34	&		12	&	\xmark	&	\xmark	\\
&	mpg\_example4		&   5	&	34	&		19	&	\xmark	&	\xmark	\\
&	mpg\_example5		&	9	&	60	&		11	&	\xmark	&	\xmark	\\
&	mpg\_guard1			&	1	&	6	&		2	&	0.43	&	0.18	\\
&	mpg\_guard2			&	1	&	6	&		2	&	0.49	&	0.19	\\
&	mpg\_guard3			&	1	&	6	&		2	&	0.46	&	0.17	\\
&	mpg\_guard4			&	1	&	6	&		2	&	0.58	&	0.18	\\
&	mpg\_ite1			&	3	&	15	&		8	&	369.57	&	361.21	\\
&	mpg\_ite2			&	5	&	29	&		11	&	\xmark	&	\xmark	\\
\cdashline{1-7}
\parbox[t]{3mm}{\multirow{19}{*}{\rotatebox[origin=c]{90}{\!\!\!\!\!\!\textsc{LimitedPlus}}}}
&	array\_sum\_2\_5	&	19	&	89	&		1	&	\xmark	&	\xmark	\\
&   $\dots$				&		&		&			&	\xmark	&	\xmark	\\
&   array\_sum\_10\_5	&	461	&	15515	&		1	&	\xmark	&	\xmark	\\
&	array\_sum\_2\_15	&	59	&	382	&		1	&	\xmark	&	\xmark	\\
&   $\dots$				&		&		&			&	\xmark	&	\xmark	\\
&   array\_sum\_8\_15	& $641$	&	$25211$	&		1	&	\xmark	&	\xmark	\\
&	mpg\_example1		&	59	&	382	&		1	&	\xmark	&	\xmark	\\
&	mpg\_example2		&	21	&	$101$	&		1	&	\xmark	&	\xmark	\\
&	mpg\_example3		&	143	&	2118	&		1	&	\xmark	&	\xmark	\\
&	mpg\_example4		&   443	&	14374	&		1	&	\xmark	&	\xmark	\\
&	mpg\_example5		&	351	&	9746	&		1	&	\xmark	&	\xmark	\\
&	mpg\_guard1			&	7	&	24	&		1	&	\xmark	&	\xmark	\\
&	mpg\_guard2			&	9	&	34	&		1	&	\xmark	&	\xmark	\\
&	mpg\_guard3			&	11	&	41	&		1	&	\xmark	&	\xmark	\\
&	mpg\_guard4			&   13	&	53	&		1	&	\xmark	&	\xmark	\\
&	mpg\_ite1			&	9	&	34	&		1	&	\xmark	&	\xmark	\\
&	mpg\_ite2			&	13	&	53	&		1	&	\xmark	&	\xmark	\\
&	mpg\_plane1			&	2	&	5	&		3	&		0.69	&	0.13	\\
&	mpg\_plane2			&	17	&	$60$	&		1	&\xmark		&	\xmark	\\
&	mpg\_plane3			&	29	&	$122$	&		1	&\xmark		&	\xmark	\\
\hline 
\end{tabular}
\label{Ta:results}
\vspace{-2mm}
\end{table*}

\begin{table*}[!t]
	\caption{
		Performance of \name on \textsc{LimitedConst} benchmarks.
		\xmark\ denotes a timeout.
	}
	\scriptsize
	\centering
	\begin{tabular}{cc|rrrrr}
		& \multirow{2}{*}{\bf Problem} & \multirow{1}{*}{\bf number of} & \multirow{1}{*}{\bf number of} & \multirow{1}{*}{\bf number of} & \multirow{1}{*}{\bf total} & \multirow{1}{*}{\bf \seahorn} \\
		&    &  {\bf nonterminals}  &  {\bf productions}  &  {\bf examples}  &  \multirow{1}{*}{\bf time (s)}  & \multirow{1}{*}{\bf time (s)}   \\
		\cline{1-7}
		\parbox[t]{3mm}{\multirow{42}{*}{\rotatebox[origin=c]{90}{\!\!\!\!\!\!\textsc{LimitedConst}}}}
		&	array\_search\_2	&	2	&	9	&		2	&	0.78	& 	0.32	\\
		&	array\_search\_3	&	2	&	10	&		3	&	1.26	& 	0.43	\\
		&	array\_search\_4	&	2	&	11	&		3	&	1.25	& 	0.22	\\
		&	array\_search\_5	&	2	&	12	&		3	&	1.01	& 	0.50	\\
		&	array\_search\_6	&	2	&	13	&		3	&	0.87	& 	0.41	\\
		&	array\_search\_7	&	2	&	14	&		3	&	0.85	& 	0.26	\\
		&	array\_search\_8	&	2	&	15	&		3	&	0.97	& 	0.36	\\
		&	array\_search\_9	&	2	&	16	&		3	&	0.70	& 	0.48	\\
		&	array\_search\_10	&	2	&	17	&		3	&	0.80	& 	0.37	\\
		&	array\_search\_11	&	2	&	18	&		3	&	1.09	& 	0.32	\\
		&	array\_search\_12	&	2	&	19	&		3	&	1.13	& 	0.25	\\
		&	array\_search\_13	&	2	&	20	&		3	&	0.73	& 	0.29	\\
		&	array\_search\_14	&	2	&	21	&		3	&	0.77	&	0.42	\\
		&	array\_search\_15	&	2	&	22	&		3	&	1.06	& 	0.23	\\
		&	array\_sum\_2\_5	&	2	&	8	&		2	&	1.30	&	0.77	\\
		&	array\_sum\_2\_15	&	2	&	8	&		2	&	1.46	&	0.83	\\
		&	array\_sum\_3\_5	&	2	&	9	&		2	&	1.31	&	0.86	\\
		&	array\_sum\_3\_15	&	2	&	9	&		2	&	1.28	&	0.75	\\
		&	array\_sum\_4\_5	&	2	&	10	&		2	&	2.52	&	0.60	\\
		&	array\_sum\_4\_15	&	2	&	10	&		2	&	1.35	&	0.56	\\
		&	array\_sum\_5\_5	&	2	&	11	&		2	&	1.41	&	0.72	\\
		&	array\_sum\_5\_15	&	2	&	11	&		2	&	1.43	&	0.44	\\
		&	array\_sum\_6\_5	&	2	&	12	&		2	&	2.37	&	0.55	\\
		&	array\_sum\_6\_15	&	2	&	12	&		2	&	1.56	&	0.70	\\
		&	array\_sum\_7\_5	&	2	&	13	&		2	&	0.76	&	0.59	\\
		&	array\_sum\_7\_15	&	2	&	13	&		2	&	1.87	&	0.78	\\
		&	array\_sum\_8\_5	&	2	&	14	&		2	&	1.33	&	0.63	\\
		&	array\_sum\_8\_15	&	2	&	14	&		2	&	1.53	&	0.67	\\
		&	array\_sum\_9\_5	&	2	&	15	&		2	&	1.50	&	0.40	\\
		&	array\_sum\_9\_15	&	2	&	15	&		2	&	1.44	&	0.79	\\
		&	array\_sum\_10\_5	&	2	&	16	&		2	&	2.29	&	0.74	\\
		&	array\_sum\_10\_15	&	2	&	16	&		2	&	0.87	&	0.43	\\
		&	mpg\_example1		&	2	&	8	&		1	&	0.36	&	0.17	\\
		&	mpg\_example2		&	2	&	9	&		4	&	0.50	&	0.30	\\
		&	mpg\_example3		&	2	&	9	&		1	&	0.57	&	0.33	\\
		&	mpg\_example4		&	2	&	10	&		1	&	0.44	&	0.16	\\
		&	mpg\_example5		&	2	&	8	&		1	&	0.99	&	0.36	\\
		&	mpg\_guard1			&	2	&	9	&		6	&	3.08	&	1.19	\\
		&	mpg\_guard2			&	2	&	9	&		4	&	2.49	&	1.35	\\
		&	mpg\_guard3			&	2	&	9	&		4	&	1.83	&	0.50	\\
		&	mpg\_guard4			&   2	&	9	&		4	&	24.18	&	21.67	\\
		&	mpg\_ite1			&	2	&	9	&		1	&	0.33	&	0.19	\\
		&	mpg\_ite2			&	2	&	9	&		1	&	0.41	&	0.25	\\
		&	mpg\_plane2			&	2	&	9	&		1	&	0.47	&	0.32	\\
		&	mpg\_plane3			&	2	&	9	&		1	& 	0.74	&	0.51	\\
		\hline 
	\end{tabular}
	\label{Ta:results2}
	\vspace{-2mm}
\end{table*}


\fi
\end{document}